\documentclass[oneside,english]{elsart}
\usepackage[T1]{fontenc}
\usepackage[latin1]{inputenc}
\usepackage[dvips]{graphicx}
\usepackage{amssymb}
\usepackage{amsmath}
\usepackage{float}
\usepackage{babel}

\theoremstyle{plain}

\theoremstyle{remark}
\newtheorem{Remark}{Remark}

\theoremstyle{definition}
\newtheorem{defi}{Definition}
\newtheorem{proof}{Proof}
\newtheorem{ex}{Example}

\newcommand{\rz}{\mathbb{R}}

\newcommand{\N}{\mathbb{N}}

\makeatother

\begin{document}

\begin{frontmatter}

\title{COMPOSITION  LAW OF  CARDINAL ORDERING PERMUTATIONS }

\author[addr1,addr2]{Jesús San Martín}
\ead{jsm@dfmf.uned.es}
\author[addr1]{Mª Jose Moscoso}
\ead{mariajose.moscoso@upm.es}

\address[addr1]{Departamento de Matemática Aplicada, E.U.I.T.I. Universidad
Politécnica de Madrid. 28012-Madrid}
\address[addr2]{Departamento de Física Matemática y de Fluidos, Facultad
de Ciencias. Universidad Nacional de Educación a Distancia. 28040-Madrid,
SPAIN.}

\author[addr3]{A. González Gómez}
\ead{antonia.gonzalez@upm.es}
\address[addr3]{Departamento de Matemática Aplicada a los Recursos Naturales,
E.T. Superior de Ingenieros de Montes. Universidad Politécnica de
Madrid. 28040-Madrid, SPAIN.}

\begin{abstract}
In this paper the theorems that determine composition laws for both cardinal ordering permutations and their inverses are proven. So, the relative positions of points in a $hs-$periodic orbit become completely known as well as in which order those points are visited. No matter how a $hs$-periodic orbit emerges, be it through a period doubling cascade $(s=2^n)$ of the $h-$periodic orbit, or as a primary window (like the saddle-node bifurcation cascade with $h=2^n$), or as a secondary window (the birth of a $s-$periodic window inside the $h-$periodic one). Certainly, period doubling cascade orbits are  particular cases with $h=2$ and $s=2^n.$   Both composition laws are also shown in algorithmic way for their easy use.
\end{abstract}


\end{frontmatter}

\section{Introduction}

As Hao Bai-Lin says in his book ~\cite{Hao} ``Symbolic dynamics  might be the only rigorous way to define  chaos...''. Certainly, this nice theory brings order into chaos labelling and ordering periodic orbits and therefore giving us the grammar of chaos.
Working in this area, Metropolis, Stein and Stein proved, in their seminal work ~\cite{metro}, that the transformations of an interval in itself, by an unimodal function $f$, generate  limit sets exhibiting universal structures. The limit sets are characterized  by finite sequences of symbols $\mathrm{R}$ and $\mathrm{L}$, that indicate whether successive iterates of the unimodal function  are plotted right or left of the critical point of the function. If a sequence of symbols is repeated periodically then it represents a periodic orbit of the dynamical system.

Once symbolic sequences associated to orbits are known a new question arises. How are symbolic sequences of $h-$periodic  and $s-$periodic  orbits  composed to obtain  the symbolic sequence of a $hs-$periodic orbit?. The answer to this problem is found in the fundamental work by Derrida, Gervois  and  Pomeau ~\cite{De}. The  Derrida-Gervois-Pomeau  composition law allows us to get new  information  about a dynamical system from information already known.  It is an example of the power of symbolic dynamics, that we want to generalize in this paper.

If we pay attention to symbolic sequence of a  $h-$periodic orbit we observe a sequence of letters $\mathrm{R}$ and $\mathrm{L}$, new questions, logically, arise: what are the relative positions of different   $\mathrm{R}(\mathrm{L})$? what $\mathrm{R}(\mathrm{L})$ is rightmost  or leftmost than  other $\mathrm{R}(\mathrm{L})$? is a given $\mathrm{R}(\mathrm{L})$ to the right or left of another previously  given $\mathrm{R}(\mathrm{L})$?.
To answer these questions we need to introduce  an order on the set of points  that form the $h-$periodic  orbit. This order determines the relative positions of points about  which  questions are answered. This problem was approached and solved in  \cite{Mar} for period doubling cascade ~\cite{Fei1} ~\cite{Fei2} (see  kneading theory \cite{Mil} and invariant coordinates \cite{Gil}) for a different approach). In this same paper, the authors met  the cardinal ordering permutation $\text{{\large{$\sigma_{2^{k}}$}}}$ for every orbit of period doubling cascade. This permutation allows to know both how many iterates   are needed to reach a given point of the orbit from another fixed point and the relative position of the point.

Bearing  in mind that the symbolic sequence composition law provides information about new orbits  then the natural generalization is how to compose two arbitrary cardinal ordering permutation $\text{{\large{$\sigma_{h}$}}}$ and $\text{{\large{$\sigma_{s}$}}}$ to obtain $\text{{\large{$\sigma_{hs.}$}}}$   From this generalization  visiting order and relative position of points in a $hs-$periodic orbit are obtained, that is, new information about the dynamical system is deduced. In fact, as it is discussed in the  conclusions, infinitely many orbits are fully determined by only two cardinal ordering permutations. Furthermore, cardinal ordering permutations  for period doubling cascade and saddle-node bifurcation cascade \cite{Mar1, Mar2} are immediately derived and therefore all underlying periodic structures
in an unimodal map. That is, these structures are fully determined from their cardinal ordering permutations that we can obtain with the composition law. The most remarkable goal of this paper is to find such a composition law.

Before starting the main part of the paper we are going to introduce some intuitive ideas and related problems that will be fundamental in its development  and useful to understand the sometimes tedious mathematical apparatus.

If the $hs-$periodic orbit is the result of the birth of a $s-$periodic orbit in a $h-$periodic  window, there is an intuitive and geometric understanding of the problem that we will use to obtain the composition law. Imagine there are $h$ rooms in a corridor, bearing $s$ chairs each. Every time you move an iterate you go into a different room and you sit in  a chair in the room. You do not visit  the rooms one after the another, as found in the corridor, but according to a rule given by   $\text{{\large{$\sigma_{h}$}}}$. Each time you visit the same room, you will sit in a different chair according to  $\text{{\large{$\sigma_{s.}$}}}$

Two questions naturally arise:
\begin{enumerate}
  \item How many movements (iterates) do we have to make to reach the $r-$th chair (supposed they were arranged in a row)?
  \item What chair would you reach if you made  $q$ movements (iterates)?
\end{enumerate}

Both questions will be answered in this paper. They are very related to composition law. That can give an idea of the power of the law. Nonetheless to get  the answers two important problems have to be previously solved.

\begin{itemize}
\item [i)] There are two kinds of rooms, associated with a maximum and minimum of $f^h$ respectively (every iteration of $f$ represents a movement). When we enter in a room of maximum kind then   $\text{{\large{$\sigma_{s}$}}}$ indicates on what chair we sit down, but if we enter in a room of minimum kind we must use the conjugate of the permutation $\text{{\large{$\overline{\sigma_{s}}$}}}$  ~\cite{Mar}. Therefore, we have to know whether $f^h$ has  a maximum or a minimum. Should $f^h$  had $5000$ extrema and we had to take the $5000$ derivatives of $f^h$ the method would be totally useless. We have to get  the information easily and without taking derivatives, so we will prove a proposition that allows us to get this information very easily.

                   \item [ii)] If we are in a room seated in a chair then we know what chair we will sit the next time  we visit  the room, because of  $\text{{\large{$\sigma_{s.}$}}}$ But, if we visit the room for the first time we do not have that information. Therefore, it is necessary to determine what chair we will sit the first time we visit a room.

                                 \end{itemize}

The  solutions of both problems are related and they are an important milestone in finding the composition law of cardinal ordering permutation.

Apart  from its theoretical value \emph{per se}, a composition law for permutations is useful wherever permutation ruling orbits are used, since composition law indicates how to compound  different elements  to get more information. Its properties are expected to underlie in permutations modelling area preserving monotone twist maps \cite{Glen}, entropy via permutations \cite{A, Am, Ami, Ban} and patterns in dynamical system \cite{Amig, Amigo}; subjects, that in turn, underlie in other fields \cite{Z}.

 This paper is organized as follows. First, most of the definitions are introduced. Second,    some lemmas and  proposition are proven to solve the mentioned problems. Thirst, composition laws for cardinal ordering permutations are obtained. Then composition laws are shown in algorithmic way. Finally the last section contains our conclusions and discussions.  Some examples will also be shown to  ease the use of theorems and algorithms by scientists and engineers.

\medskip{}


\section{Definitions and notation}

Let $f:\mathrm{I}\subset\rz\to\mathrm{I}$ be an unimodal map,
and let $\mathrm{C}$ denote its critical point (we suppose, without loss of generality, that $\mathrm{C}$ is a maximum). Let $O_p=\{ x_{1},\ldots,x_{p}\equiv \mathrm{C}\}=\{ f(\mathrm{C}),\ldots, f^{p-1}(\mathrm{C}),f^p(\mathrm{C})\equiv \mathrm{C}\}$
be a $p$-periodic supercycle of $f$, then the first
and the second iterates of $\mathrm{C}$ determine the subinterval $\mathrm{J}=[f^{2}(\mathrm{C}),f(\mathrm{C})]$ such that $O_{p}\subset[f^{2}(\mathrm{C}),f(\mathrm{C})]=\mathrm{J}.$ We denote $\mathrm{J}_{\mathrm{L}}=[f^{2}(\mathrm{C}),\mathrm{C}]$ and $\mathrm{J}_{\mathrm{R}}=[\mathrm{C}, f(\mathrm{C})].$
Since the orbit $O_{p}$ does not have the natural order within the
interval $\mathrm{J}=[f^{2}(\mathrm{C}),f(\mathrm{C})]$, some previous definitions, introduced in \cite{Mar}, are needed.

 The set $\{ \mathrm{C}_{(1,p)}^{*},\, \mathrm{C}_{(2,p)}^{*},\ldots,\mathrm{C}_{(p,p)}^{*}\}$
will denote the descending cardinality ordering of the orbit
$O_p=\{ x_{1},x_{2},\ldots,x_{p}\equiv \mathrm{C}\}$. Given that $\mathrm{C}_{(i,p)}^{*}\ i=1,\ldots,p$, is the point of
the orbit $O_{p}$ located in the  position $'i'$ according to the descending cardinality   ordering, the point $\mathrm{C}_{(i,p)}^{*}$ is defined as the
$i-$th cardinal of the $p-$period orbit.
Notice that $f^{2}(\mathrm{C})=\mathrm{C}_{(p,p)}^{*}<\ldots<\mathrm{C}_{(2,p)}^{*}<\mathrm{C}_{(1,p)}^{*}=f(\mathrm{C}).$

 Let $O_p$
be a $p$-periodic supercycle  of $f$, we denote as {\large {${\sigma}_{p}$}} the permutation   {\large $\sigma_{p}$} $=(\sigma{(1,p)},\ldots,\sigma{(p,p)})$,  where  ${\sigma}{(i,p)}$
\, is the number of iterations of $f$ from $\mathrm{C} $ such that $f^{\sigma{(i,p)}}(\mathrm{C})=\mathrm{C}_{(i,p)}^{*}=x_{\sigma{(i,p)}}$ for all  $i=1,\ldots,p$ (see figures \ref{f1} and \ref{fig2}). This permutation  will be called the cardinal ordering permutation of the orbit $O_{p}$.

\begin{figure}
\includegraphics[width=0.86\textwidth]{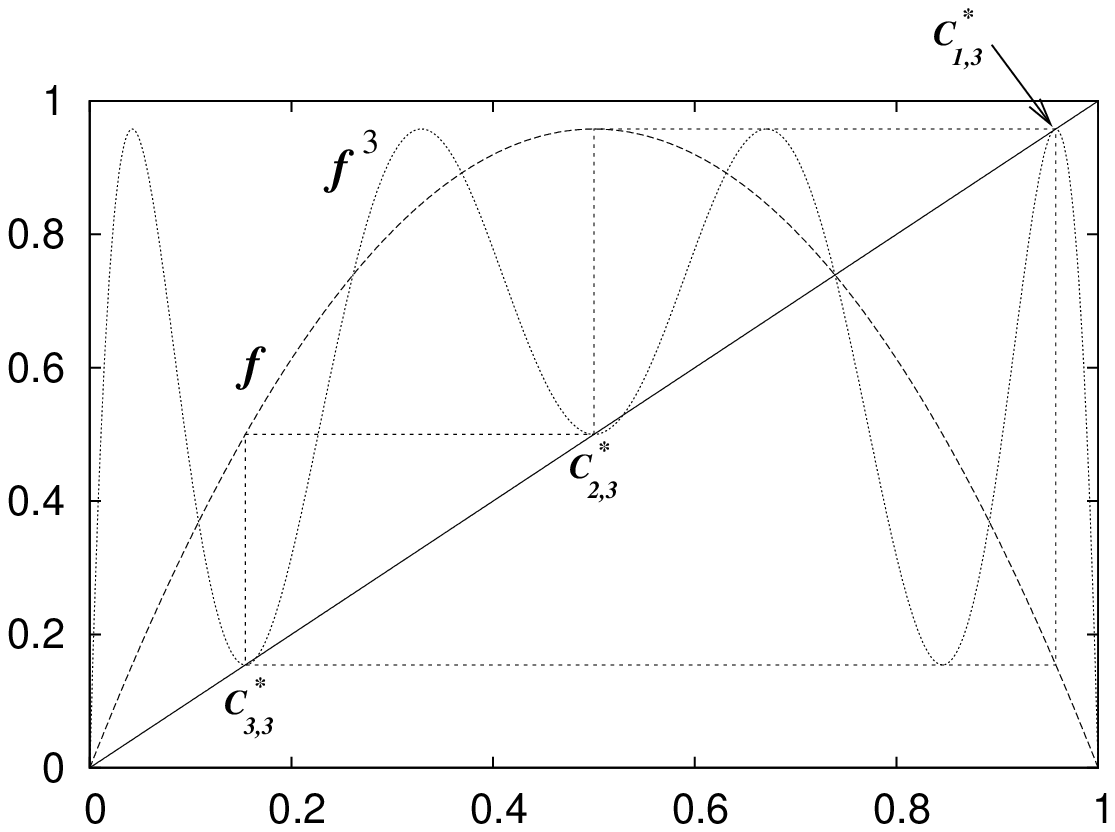}
\caption{\label{f1}%
  \(\protect \quad  \text{\large{$\sigma$}}_{3}=\left( \protect\begin{array}{ccc} 1 & 2 & 3  \protect\\  1 & 3 & 2 \protect\end{array} \right) \protect \ \)
\protect\begin{tabular}{c}
 $\leftarrow \text{indicates the position } 'i' \text{ of } \mathrm{C}_{(i,3)}^{*}$ \protect\\
$\leftarrow  \text{number of iterations, }\sigma_{(i,3)},\text{ to reach \mbox{$\mathrm{C}_{{(i,3)}}^{*}$}}$
\protect\end{tabular}
$ \protect\text{\large{$\sigma$}}^{-1}_{3}=\left( \protect\begin{array}{ccc} 1 & 2 & 3  \protect\\  1 & 3 & 2 \protect\end{array} \right) $}

\end{figure}

\begin{figure}%
\includegraphics[width=0.86\textwidth]{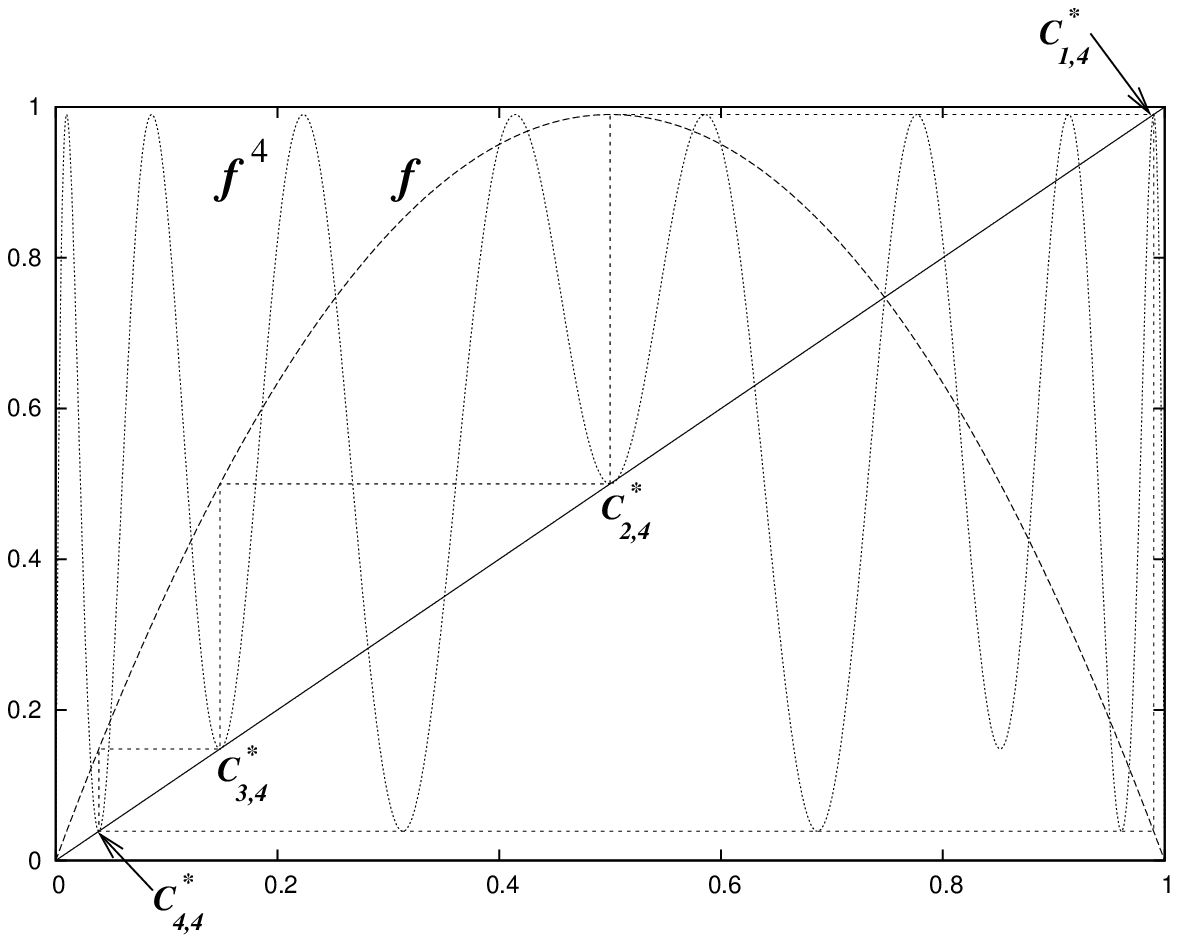}
\caption{\label{fig2}
  \(\protect \quad \text{\large{$\sigma$}}_{4}=\left( \protect\begin{array}{cccc} 1 & 2 & 3 &4 \protect\\  1 & 4 & 3&2 \protect\end{array} \right) \protect \qquad\)
\( \text{\large{$\sigma$}}^{-1}_{4}=\left( \protect\begin{array}{cccc} 1 & 2 & 3 &4 \protect\\  1 & 4 & 3&2 \protect\end{array} \right) \)
}
\end{figure}

Authors proved in \cite{Mar}  that  the inverse permutation determines  the order  in which  the orbit  is visited. That is, the visiting sequence of the $p$-periodic supercycle  $O_{p}$ from $\mathrm{C}$  is the inverse permutation {\large{${\sigma}$}}$_{p}^{-1}=(\sigma^{-1}{(1,{p})},\,\sigma^{-1}{(2,{p})},\ldots,\sigma^{-1}{({p}, {p})})$  (see figures \ref{f1} and \ref{fig2}). Therefore, depending on which context  it is used,  {\large{${\sigma}$}}$_{p}^{-1}$ will be called inverse permutation  or visiting sequence.

\begin{defi}\label{def1}
 The symbolic sequence of the supercycle $O_{p}$, denoted by $\mathrm{I}_1, \mathrm{I}_2,\ldots,\mathrm{ I}_{p-1}, \mathrm{C}$ where $\mathrm{I}_{i}=\mathrm{R}$ or $\mathrm{L}$ \, for all $i=1,\ldots,p-1$, will be called  $p-$symbolic sequence.
\end{defi}

\begin{defi}\label{def2} Let  $\mathrm{I}_1, \mathrm{I}_2,\ldots, \mathrm{I}_{p-1}, \mathrm{C}$ be a $p-$symbolic sequence. Then
\begin{enumerate}
  \item [(i)] $\mathrm{I}_j$  is the $p-$symbolic sequence element associated  to the element $x_j$ of the orbit $O_p$.
  \item [(ii)] $\mathrm{I}_1, \mathrm{I}_2,\ldots,\mathrm{I}_{j-1}$  is the subsequence preceding  $\mathrm{I}_j$ of the $p-$symbolic sequence. The subsequence preceding  $I_1$ of the $p-$symbolic sequence is $\mathrm{C}$ by definition.
  \item [(iii)] $\mathrm{I}_j,\ldots,\mathrm{I}_{p-1},\mathrm{ C}$  is the subsequence following   $\mathrm{I}_{j-1}$ of the $p-$symbolic sequence.
  \item [(iv)] $\mathrm{I}_{j}, \mathrm{I}_{j+1}, \ldots, \mathrm{I}_{q-1}, \mathrm{I}_q$ is the $p-$symbolic sequence from  $\mathrm{I}_j$ to $\mathrm{I}_{q}$.
\end{enumerate}
\end{defi}

We want to study  a  $hs$-periodic supercycle $O_{hs}$ of $f$ with $hs$ points, resulting from the combination of  two  supercycles   with  $h$ and  $s$ points respectively. In order to get that we split the $hs$ points in $h$ sets, that  we will name boxes, with $s$ points each.
 The boxes are visited according to {\large{$\sigma$}}$_{h}^{-1}$, while {\large{$\sigma$}}$_{s}^{-1}$ gives the order in which points inside the boxes are visited.  As the boxes are visited
one after another according {\large{$\sigma$}}$_{h}^{-1}$, it results that the $h-1$ first elements of the symbolic sequence of $O_{hs}$  coincide with
 the  elements  of the symbolic sequence of $O_{h}$.
The boxes are associated with the extrema of $f^h$ (closest to the straight line $y=x$) and each of them contains the extrema of $f^s$   (closest to the straight line $y=x$) as it can be seen in figure \ref{fighs}. Notice that the graph of $f^s$ is repeated in each extremum of $f^h$, and furthermore they intersect  the straight line $y=x$:  the points of the supercycle $O_{hs}$.

\begin{figure}
\includegraphics[width=0.86\textwidth]{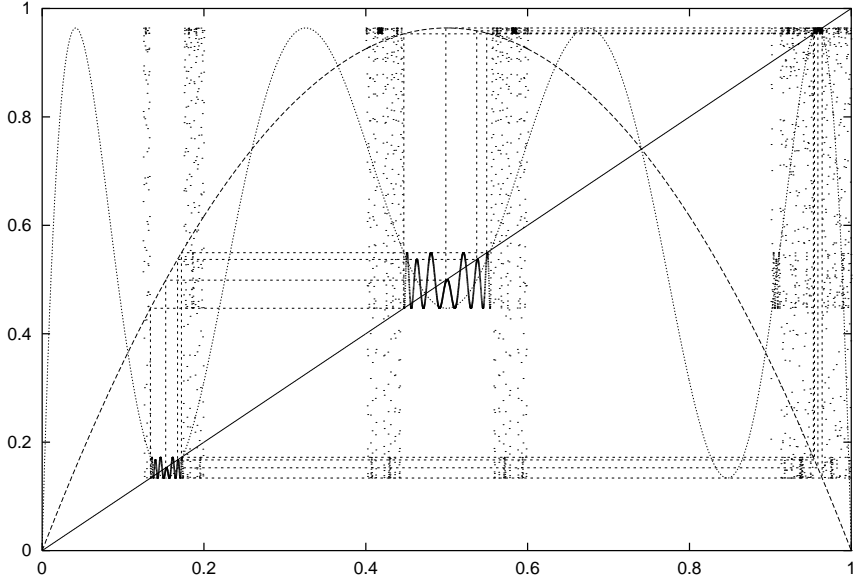}
\caption{\label{fighs} The graph of $f^s$ is repeated close to the straight line $y=x$}
\end{figure}

According to this idea we introduce the following definition

\begin{defi}\label{def3} Let $f(x;r)$ be an unimodal map such that for $r=r_h$ and $r=r_s$ it has the supercycles $O_h$ and $O_s$ having periods $h$ and $s$, respectively. We define  $hs-$periodic supercycle of $f$, denoted by  $O_{hs}$, the supercycle of $f$  $O_{hs}=\{x_1, x_2, \ldots, x_{hs}\equiv \mathrm{C}\}$ originated at $r=r_{hs}$, where:
\begin{enumerate}
  \item [i)]  The $hs$ points split in $h$ sets with $s$ points each.
  \item [ii)] The $h$ sets are in the neighborhoods of the extrema of $f^h(x;r_{hs})$, closest to the straight line $y=x$, containing  the graph of $f^s(x;r_s)$.
\end{enumerate}
\end{defi}

The definition \ref{def3} refers to supercycles associated with period doubling cascade, saddle-node bifurcation cascade and the birth of windows inside other windows.

\begin{defi}\label{def4}  Let $\{ \mathrm{C}_{(1,hs)}^{*},\, \mathrm{C}_{(2,hs)}^{*},\ldots,\mathrm{C}_{(hs,hs)}^{*}\}$ be  the cardinal points of  $O_{hs}$.  The $i-$th cardinal  box of  $O_{hs}$ with $i=1,\ldots, h$,   denoted by  $H_{i}^*$, is the set
$H_{i}^*=\{c_{(1,i)}^*, c_{(2,i)}^*, \ldots, c_{(s,i)}^*\}$,
where
\[c_{(j,i)}^*=\mathrm{C}_{((i-1)s+j,hs)}^* \qquad \forall \, j=1,\ldots,s\]
\end{defi}

Observe that elements belonging  $O_{hs}$ are denoted by capital letters, but when these elements belong to a box they are denoted by small letters. With this definition we achieve that both   the  boxes $H_{i}^*$  with $i=1,\ldots, h$ and the points inside  them  maintain the descending cardinality  ordering.

 \begin{Remark}\label{r2} As said above $\sigma_{h}^{-1}$  determines how boxes are visited, therefore  if $x_j\in O_{hs}$ with $1\leq j \leq h$ then $x_j\in H_{{\sigma^{-1}{(j,h)}}}^*$. In fact,  $x_{j+rh}$ with $r=0,\ldots, s-1$ are the unique points  belonging to  $ H_{{\sigma^{-1}{(j,h)}}}^*$. In other words we can  say that  if   $x_{j}\in H^*_{i}$  then $j={\sigma{(i,h)}}+ rh$ for some $r=0,\ldots, s-1.$\end{Remark}

 \begin{defi}\label{def5}  The central    box of $O_{hs}$, denoted as ${H^*_{\mathrm{C}}}$, is the box containing $\mathrm{C}$, that is ${H^*_{\mathrm{C}}}={H^*}_{{\sigma^{-1}_{(h,h)}}}$. Notice that, according  to remark \ref{r2},  $x_{rh}\in H^*_{\mathrm{C}}=H^*_{{\sigma^{-1}_{(h,h)}}}$  with $r=1,\ldots, s$.
\end{defi}

We must point out that   $H^*_{\sigma^{-1}_{(j,h)}}\neq H^*_{\mathrm{C}}$ with $1\leq j<h$, then either  $H^*_{\sigma^{-1}_{(j,h)}}\subset \mathrm{J}_{\mathrm{L}}$  or $H^*_{\sigma^{-1}_{(j,h)}}\subset {\mathrm{J}}_{\mathrm{R}}$. This  result  will be very useful and widely used along the paper.

\begin{defi}\label{def6}   The box  $H^*_{i}$ is a convex box  of $O_{hs}$ if there  exists some $x_k\in H^*_{i}\subset O_{hs}$ such that $x_k$ is a maximum of $f^h$. The box  $H^*_{i}$ is a concave box  of $O_{hs}$
if there  exists some $x_k\in H^*_{i}\subset O_{hs}$ such that $x_k$ is a minimum of $f^h$.
\end{defi}

\begin{defi}\label{def7}   We say  that $c^*_{(1,i)}$  and  $c^*_{(s,i)}$ are boundary points for all $i=1,\ldots,h$.
\end{defi}

 \begin{defi}\label{def8}   Let be $O_{hs}=\{x_1, x_2, \ldots, x_{hs}\equiv \mathrm{C}\}$. We say   that the order entry from $H^*_{\sigma^{-1}(n,h)}$ to  $H^*_{\sigma^{-1}(m,h)}$, with $1\leq n< m\leq h$,  is conserved if  $f^{m-n}(c^*_{(1,\sigma^{-1}(n,h))})=c^*_{(1,\sigma^{-1}(m,h))}$ or  $f^{m-n}(c^*_{(s,\sigma^{-1}(n,h))})=c^*_{(s,\sigma^{-1}(m,h))}$.

We say  that  the order entry from $H^*_{\sigma^{-1}(n,h)}$ to  $H^*_{\sigma^{-1}(m,h)}$, with $1\leq n< m\leq h$,  is reversed if  $f^{m-n}(c^*_{(1,\sigma^{-1}(n,h))})=c^*_{(s,\sigma^{-1}(m,h))}$ or  $f^{m-n}(c^*_{(s,\sigma^{-1}(n,h))})=c^*_{(1,\sigma^{-1}(m,h))}$.
\end{defi}


\section{Theorems}

In order to get our goals  we need to prove some lemmas that  make easy the  proof of the theorems.
\begin{lem}\label{lema1} Let $O_{hs}=\{x_1, x_2, \ldots, x_{hs}\equiv C\}$ be a  $hs-$periodic supercycle of $f$. Then $x_1,x_2,\ldots, x_h$ are boundary points.
\end{lem}
\begin{proof} As $x_1=\mathrm{C}_{(1,hs)}^*=c_{(1,1)}^* $ and $x_2=\mathrm{C}_{(hs,hs)}^* =c_{(s,h)}^*$ it results that $x_1$ and $x_2$ are boundary points from definition \ref{def7}.

Let us first prove that $x_3$ is also   a boundary points.
We have $x_2=c^*_{(s,h)}$ then  $x_2\in  H^*_{h}=H^*_{\sigma^{-1}(2,h)}$ thereby  $x_2< x_{2+rh}$ for all $r=1,\ldots,s-1$.  Given $H^*_{h}\subset J_{L}$ and $f$ is an increasing function in $J_{L}$ it results that  $x_3=f(x_2)<  f(x_{2+rh})$. As {\large{$\sigma_{h}^{-1}$}} determines how boxes are visited  we obtain  that $x_3,  f(x_{2+rh})\in  H^*_{\sigma^{-1}(3,h)}$ and as  $x_3<  f(x_{2+rh})$  it yields  $x_3=c^*_{(s,\sigma^{-1}(3,h))}$ and from definition \ref{def7}  it results $x_3$ is a boundary point.

Let us next  prove that $x_4$ is also   a boundary point. As  $x_3=c^*_{(s,\sigma^{-1}(3,h))}\in H^*_{\sigma^{-1}(3,h)}$ then we have $x_3< x_{3+rh}$ for all $r=1,\ldots,s-1$.
 Two cases are possible:
 \begin{enumerate}
   \item [i)] $H^*_{\sigma^{-1}(3,h)}\subset J_{L}$.

   In this case  we repeat the previous argument  and obtain   $x_4=c^*_{(s,\sigma^{-1}(4,h))}$.
   \item [ii)] $H^*_{\sigma^{-1}(3,h)}\subset J_{R}$.

    As $f$ is  a decreasing function in $J_R$  it results $f(x_3)> f(x_{3+rh})$ for all $r=1,\ldots,s-1$. Since {\large{$\sigma_{h}^{-1}$}} determines how boxes are visited and furthermore  $f(x_3)=x_4, f(x_{3+rh})\in  H^*_{\sigma^{-1}(4,h)}$ then  it yields that $x_4=c^*_{(1,\sigma^{-1}(4,h))}$ and, from definition \ref{def7},  $x_4$ is a boundary point.
 \end{enumerate}

The same proof works for  $x_5, \ldots, x_{h}.$ \end{proof}

\begin{figure}
\includegraphics[width=1.15\textwidth]{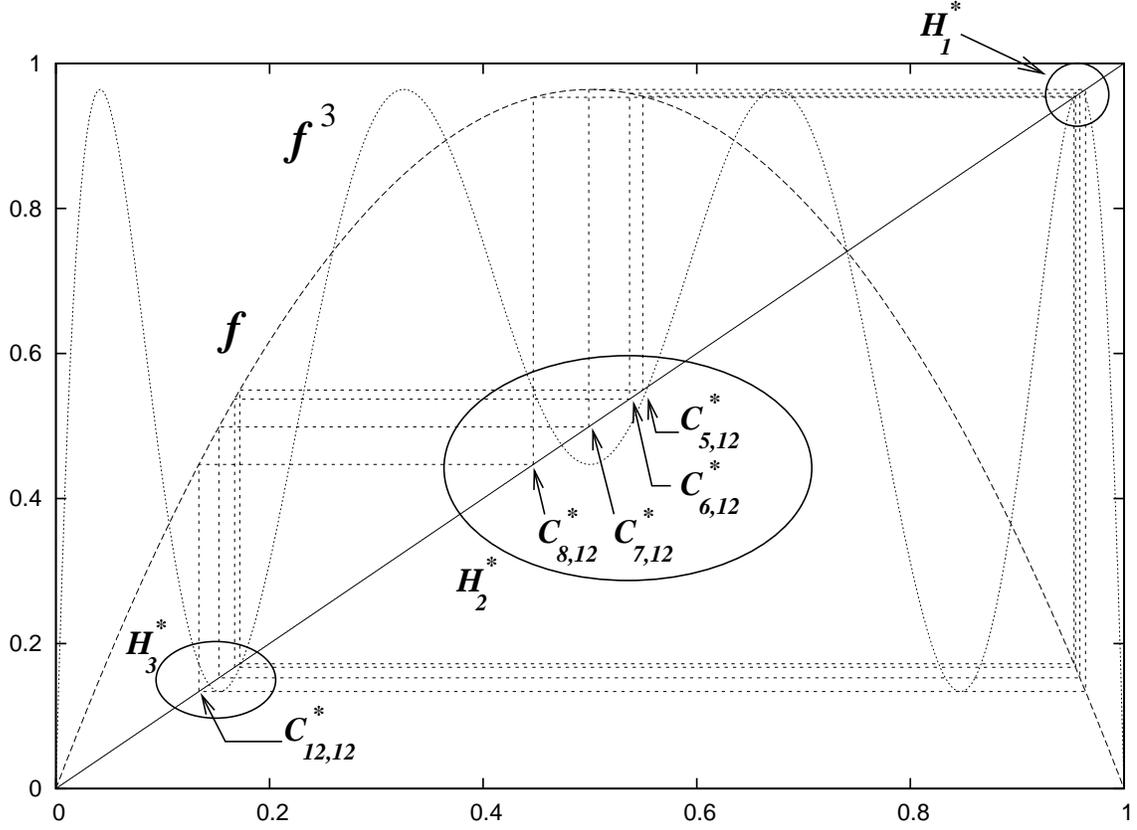}\\
\caption{ \label{f3} An $O_{3 \cdot 4}$ supercycle can be observed. Boxes $H^{*}_{1}$, $H^{*}_{2}$, $H^{*}_{3}$ have been plotted in the extrema of $f^{3}$. Some cardinals have also been plotted. $\mathrm{C}^{*}_{5,12}$ and $\mathrm{C}^{*}_{8,12}$ are boundary points in $H^{*}_{2}$.}
\end{figure}
\begin{figure}
\includegraphics[width=0.85\textwidth]{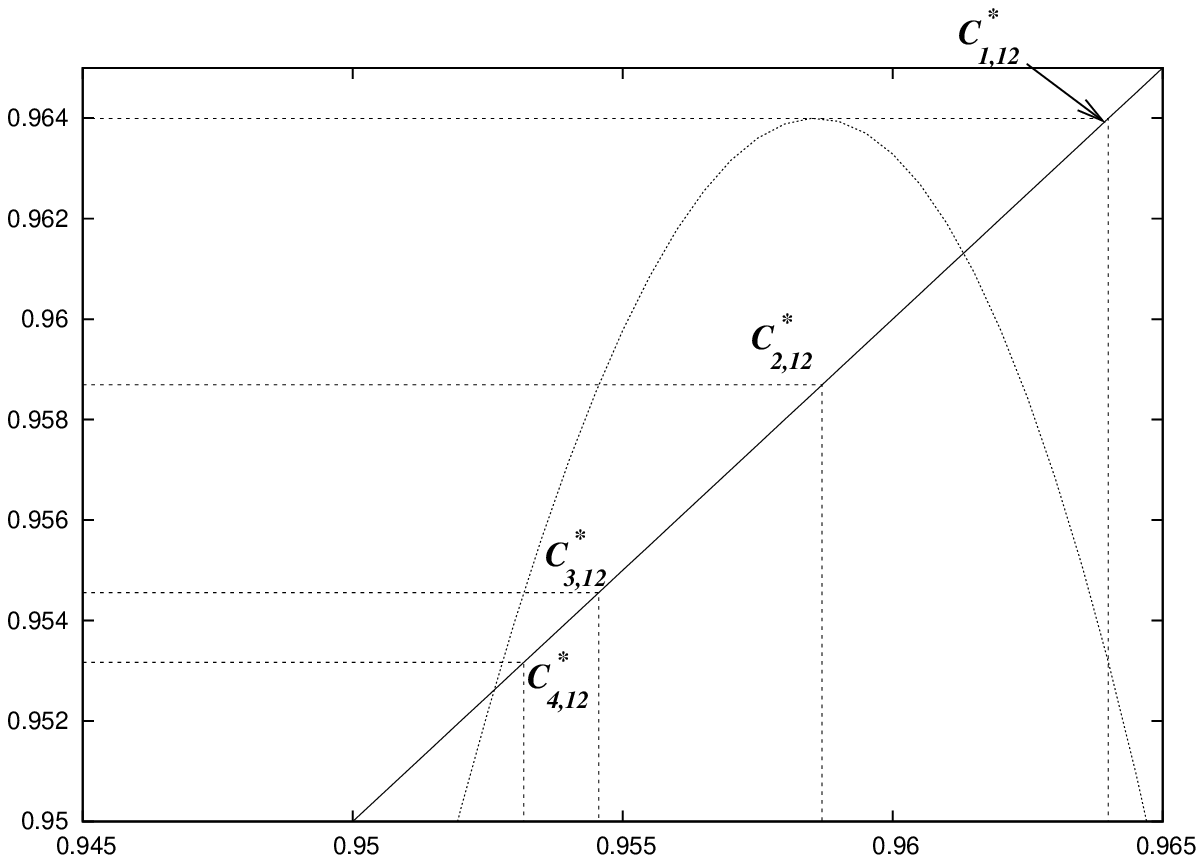}
\caption{\label{figure4} Enlargement of the box $H^{*}_{1}$ in figure \ref{f3}, where boundary points $C^{*}_{1,12}$ and $C^{*}_{4,12}$ can be seen}
\end{figure}

\begin{ex} If we pay attention to figures \ref{f3} and \ref{figure4}, where a $O_{3\cdot 4}$ supercycle is plotted, we can see that $x_1=f(\mathrm{C}), \, x_2=f(x_1)$ and $x_3=f(x_2)$ are boundary points. They are spotted in the  boundaries of their respective boxes.
\end{ex}

In order to prove the next lemma we want to point out  that if $1\leq n<m\leq h$,  $n, m\in \N$, then by the above lemma we will only work with boundary points; hence, ``entry order is not conserved'' is equivalent to ``entry order is reversed''.

\begin{lem}\label{lema2} Assuming  $1\leq n<m\leq h$, \ $n, m,h\in \N$, the entry order from $H^*_{\sigma^{-1}(n,h)}$ to  $H^*_{\sigma^{-1}(m,h)}$ is conserved if only if the $\mathrm{R}-$parity of the $hs-$symbolic sequence from $\mathrm{I}_n$ to $\mathrm{I}_{m-1}$ is even.
\end{lem}

\begin{proof}
  Suppose the $\mathrm{R}-$parity of the $hs-$symbolic sequence from $\mathrm{I}_n$ to $\mathrm{I}_{m-1}$ is even.

Let  $x_n\in H^*_{\sigma^{-1}(n,h)}$  with $1\leq n<m\leq h$. According to lemma  \ref{lema1}, $x_n$ is a boundary point. Without loss of generality, let $x_n=c^*_{(1,\sigma^{-1}{(n,h)})}$  and  $\mathrm{I}_n$  the  $hs-$symbolic sequence element associated to $x_n$.

 It is deduced from  proof of lemma  \ref{lema1} that

\[f(c^*_{(1,\sigma^{-1}{(n,h)})})=\left\{
    \begin{array}{lll}
      c^*_{(1,\sigma^{-1}{(n+1,h)})} & &\hbox{if }   H^*_{\sigma^{-1}(n,h)}\subset J_\mathrm{L}\\
\\
      c^*_{(s,\sigma^{-1}{(n+1,h)})} & &\hbox{if }   H^*_{\sigma^{-1}(n,h)}\subset J_\mathrm{R}
    \end{array}
  \right.
   \]\[ f(c^*_{(s,\sigma^{-1}{(n,h)})})=\left\{
    \begin{array}{lll}
      c^*_{(s,\sigma^{-1}{(n+1,h)})} & &\hbox{if }   H^*_{\sigma^{-1}(n,h)}\subset J_\mathrm{L}\\
\\
      c^*_{(1,\sigma^{-1}{(n+1,h)})} & &\hbox{if }   H^*_{\sigma^{-1}(n,h)}\subset J_\mathrm{R}
    \end{array}
  \right.
\]
 As   $\mathrm{I}_n=\mathrm{L}$ if and only if    $H^*_{\sigma^{-1}(n,h)}\subset \mathrm{J}_\mathrm{L}$ and    $\mathrm{I}_n=\mathrm{R}$ if and only if     $H^*_{\sigma^{-1}(n,h)}\subset \mathrm{J}_\mathrm{R}$ (meanwhile  $H^*_{\sigma^{-1}(n,h)}$ is not the central box  because the central box is the last visited one), we  have that

\[f(c^*_{(1,\sigma^{-1}{(n,h)})})=\left\{
    \begin{array}{lll}
      c^*_{(1,\sigma^{-1}{(n+1,h)})} & &\hbox{if }   \mathrm{I}_n=\mathrm{L}\\
\\
      c^*_{(s,\sigma^{-1}{(n+1,h)})} & &\hbox{if }   \mathrm{I}_n=\mathrm{R}
    \end{array}
  \right.
\]

\[f(c^*_{(s,\sigma^{-1}{(n,h)})})=\left\{
    \begin{array}{lll}
      c^*_{(s,\sigma^{-1}{(n+1,h)})} & &\hbox{if }   \mathrm{I}_n=\mathrm{L}\\
\\
      c^*_{(1,\sigma^{-1}{(n+1,h)})} & &\hbox{if }   \mathrm{I}_n=\mathrm{R}
    \end{array}
  \right.
\]
Thus, if $\mathrm{I}_n=\mathrm{L}$ the entry order from $H^*_{\sigma^{-1}(n,h)}$ to  $H^*_{\sigma^{-1}(n+1,h)}$ is conserved; whereas  that if $\mathrm{I}_n=\mathrm{R}$ the entry order is reversed.

Using the previous argument  it successively results that if the $\mathrm{R}-$parity of the $hs-$symbolic sequence from $\mathrm{I}_n$ to $\mathrm{I}_{m-1}$ is even then
$f^{m-n}(c^*_{(1,\sigma^{-1}(n,h))})=c^*_{(1,\sigma^{-1}(m,h))}$, that is, the entry order from $H^*_{\sigma^{-1}(n,h)}$ to  $H^*_{\sigma^{-1}(m,h)}$ is conserved and necessary condition has been proved.

  On the other hand if  the $\mathrm{R}-$parity of the $hs-$symbolic sequence from $\mathrm{I}_n$ to $\mathrm{I}_{m-1}$ is odd then, by using the above argument it is concluded that
$f^{m-n}(c^*_{(1,\sigma^{-1}(n,h))})=c^*_{(s,\sigma^{-1}(m,h))}$, that is, the entry order from $H^*_{\sigma^{-1}(n,h)}$ to  $H^*_{\sigma^{-1}(m,h)}$ is reversed. Note that we have actually proved the sufficiency  condition.

 \end{proof}

\begin{ex}
Let us pay attention to figure  \ref{f3}. The $\mathrm{R}-$parity from $\mathrm{I_1}$ (associated to $ x_1=\mathrm{C}^*_{(1,12)})$ to $\mathrm{I_2}$ (associated to $x_{2}=\mathrm{C}^*_{(12,12)})$ is odd therefore according to lemma \ref{lema2} the entry   order is reversed, as we can see in figure  \ref{f3}, paying attention to boxes $H_1$ and $H_3$.
The $\mathrm{R}-$parity from $\mathrm{I_2}$ (associated to $x_{2}=\mathrm{C}^*_{(12,12)})$  to $\mathrm{I_3}$ (associated to $x_{3}=\mathrm{C}^*_{(8,12)})$  is even therefore , according to lemma \ref{lema2}, the entry   order is conserved, as we  can  see in figure  \ref{f3}, paying attention to boxes $H_3$ and $H_2$.
\end{ex}

The lemma, we have just proven, allows us to know what point in a box is reached the first time that box is visited. The following times we visit that box we will know what point  is reached by using  {\large $\sigma$}$^{-1}_{s}$ or $\overline{\text{\large$\sigma$}^{-1}_{s}}$ depending on whether the box is associated with a maximum or a minimum of $f^h$. This kind of information will be obtained in  lemma \ref{lema3} for which we need to introduce the following definition.

\begin{defi} Let $\mathrm{I}_1, \ldots, \mathrm{I}_{p-1}, \mathrm{C}$ be a  $p-$symbolic sequence. We  define the sign of $p-$symbolic sequence, and we denote as $\text{sgn}(\mathrm{I}_1, \ldots, \mathrm{I}_{p-1}, \mathrm{C})$, as the number given by
\[\text{sgn}(\mathrm{I}_1, \ldots, \mathrm{I}_{p-1}, \mathrm{C})=\left\{
                                                    \begin{array}{lc}
                                                      1 & \hbox{if $\mathrm{R}-$parity  of $\mathrm{I}_1, \ldots, \mathrm{I}_{p-1}$ is even} \\
                                                      -1 & \hbox{otherwise}
                                                    \end{array}
                                                  \right.
\]

\end{defi}
\begin{lem}\label{lema3} Let $O_{p}=\{x_1, \ldots, x_{p}\equiv \mathrm{C}\}$ be a $p$-periodic  supercycle   of an unimodal function $f$ with a maximum at $C$. Then $f^p$ has a maximum (minimum) at $x_i$, for all $i=1,\ldots, p$, if and only if the  $\mathrm{R}-$parity of the subsequence preceding $\mathrm{I}_i$ in the  $p-$symbolic sequence is even (odd).
\end{lem}

\begin{proof}
 As we are  working with a supercycle,  then $(f^p(x_i))'=0$ for all $i=1,\ldots, p$ and all the points belonging  to the supercycle are either  maxima or minima. Thus, we only need to study the sign of $(f^p(x_i))''$ at each $x_i\in O_p$ to know whether we have a maximum or a minimum.

By using the chain rule and taking into account that $(f^{p-i})'(x_i)=f'(\mathrm{C})=0$ it yields:

\[(f^p(x_i))''=f''(f^{p-i}(x_i))\prod_{j=1}^{i-1}f'(f^{p-j}(x_i))\prod_{j=i+1}^{p} \left(f'(f^{p-j}(x_i))\right)^2 \] As $f$ has a maximum at  $\mathrm{C}$ it follows $f''(f^{p-i}(x_i))=f''(\mathrm{C})< 0$ and it implies

 \[\text{sgn} \left((f^p(x_i))''\right)=- \text{sgn} \left(\prod_{j=1}^{i-1}f'(f^{p-j}(x_i))\right)\]consequently,
  \[\text{sgn} \left((f^p(x_i))''\right)=- \text{sgn} \left(\prod_{j=1}^{i-1}f'(x_j)\right)\]

  Since  $f'(x_i)>0 $ if  $\mathrm{I}_i=\mathrm{L}$ and  $f'(x_i)<0,$ if  $\mathrm{I}_i
  =\mathrm{R}$, then
   \[\text{sgn} \left(\prod_{j=1}^{i-1}f'(x_j)\right)=\text{sgn} \left(\mathrm{I}_1, \ldots, \mathrm{I}_{i-1},\mathrm{C}\right)\]
    thus, we obtain

\begin{equation}\label{ec1}
   \text{sgn} \left((f^p(x_i))''\right)=- \text{sgn} \left(\mathrm{I}_1, \ldots, \mathrm{I_{i-1}}, \mathrm{C} \right)
\end{equation}

So  if  $f^p$ has a maximum  at $x_i$ then $\text{sgn} \left((f^p(x_i))''\right)=-1$. Therefore,  from (\ref{ec1}) we obtain

\[-1=- \text{sgn} \left(\mathrm{I}_1, \ldots, \mathrm{I_{i-1}, \mathrm{C}} \right)\]
Consequently,  the $\mathrm{R}-$parity of the subsequence preceding $\mathrm{I}_i$ is even.

On the other hand,  if the   the $\mathrm{R}-$parity of the subsequence preceding  $\mathrm{I}_i$ is even then   from (\ref{ec1}) we have
\[\text{sgn} \left((f^p(x_i))''\right)=-1 \]
 and $f^p$ has a maximum  at $x_i$.

 Notice that  we can  analogously  prove  that   $f^p$ has a minimum  at $x_i$  if and only if the $\mathrm{R}-$parity of the subsequence preceding  $\mathrm{I}_i$ is odd.

\end{proof}

\begin{ex}
In figure \ref{f1} can be seen  a $3-$periodic supercycle with  symbolic sequence $\mathrm{RLC}$. The $\mathrm{R}-$parity of the subsequence preceding $x_1$ is even, and  $f^3$ has a maximum at this point. Meanwhile the  $\mathrm{R}-$parity of the subsequences preceding $x_2$ and $x_3$ are odd and $f^3$ has  minima at these points.   \end{ex}

Lemma \ref{lema3} allows us  to know whether  a box is a concave or a convex one, without deriving $f^h$. We need to know this information because depending on it {\large $\sigma$}$^{-1}_{s}$ or  $\overline{\text{\large$\sigma$}^{-1}_{s}}$ is used to jump from one point to another inside a box. Obviously  it takes less time to calculate a previous sequence that deriving $f^h$ twice, but if $h$ is big then calculating a preceding subsequence can take a long time too. However former lemmas can be exploit to get  a more refined information as the below proposition means:

 ``A box is convex (concave) if and only if the first time the box is visited a $c^*_{(1,-)}$ ($c^*_{(s,-)}$) point is reached''.  That is, the box is convex or concave depending on whether the first time you visit the box you reach the first or  the last boundary point. That in itself hides an unexpected link between analysis and topology in dynamical systems that it must be researched.

Before proving  the proposition we must take into account the following remark.
\begin{Remark}\label{R1} Let $O_{hs}=\{x_1, x_2, \ldots, x_{hs}\equiv \mathrm{C}\}$ be a $hs-$periodic supercycle of $f$, then  $(f^{hs})'(x_i)=0$ for each  $i=1,\ldots, hs$. As $f$ is an unimodal function with an unique critical point at $\mathrm{C}$, it is deduced that there exists  $\{x'_1, x'_2, \ldots, x'_{h}\equiv \mathrm{C}\}\subset O_{hs}$ such that $(f^h)'(x'_j)=0$ for each  $j=1,\ldots, h$. These points determine the   boxes according to definition \ref{def6}.

If we chose any point belonging to the central box after $h-1$ iterates it would generate a symbolic sequence coinciding with the $h-1$ first letters of the symbolic sequence of $O_{h}$, that also coincides with the $h-1$ first letters of the  symbolic sequence of $O_{hs}.$  That is so because each box of $O_{hs}$ is visited according to {\large $\sigma$}$^{-1}_{h}$, the permutation that determines the  supercycle  $O_h$ and its  symbolic sequence.\end{Remark}

\begin{prop}\label{p1}  Let $O_{hs}=\{x_1, x_2, \ldots, x_{hs}\equiv \mathrm{C}\}$ be a $hs-$periodic supercycle of $f$ and let  $x_r\in H^*_{\sigma^{-1}(r,h)}$ with  $1\leq r\leq h$, then
\begin{enumerate}
  \item [(i)] $H^*_{\sigma^{-1}(r,h)}$ is a convex box if and only if $x_r=c^*_{(1,\sigma^{-1}(r,h))}$
  \item  [(ii)] $H^*_{\sigma^{-1}(r,h)}$ is a concave  box if and only if $x_r=c^*_{(s,\sigma^{-1}(r,h))}$
\end{enumerate}
\end{prop}
\begin{proof}
\begin{enumerate}
  \item [(i)] Let us suppose that  $H^*_{\sigma^{-1}(r,h)}$ is a convex box, then there  exists  some $x_k\in O_{hs}$ with $x_k\in H^*_{\sigma^{-1}(r,h)}$ such that
   $f^h$ has a maximum at $x_k$.

   Let  $O_{h}=\{x'_1,  \ldots, x'_{h}\equiv \mathrm{C}\}$  be the $h-$periodic  supercycle  of $f$, as $x_k\in H^*_{\sigma^{-1}(r,h)}$, according to remark \ref{R1} it must be $x_k=x'_r$. As there is a maximum at  $x_k=x'_r$  it yields, by using lemma \ref{lema3}, that the
    $\mathrm{R}-$parity of  $\mathrm{I}'_1, \ldots, \mathrm{I}'_{r-1}$ is even. Since  $\mathrm{I}'_1, \ldots, \mathrm{I}'_{r-1}$ coincides with the first $r-1$ symbols of the $hs-$symbolic sequence, that is, $\mathrm{I}_1, \ldots, \mathrm{I}_{r-1}$, then  the $\mathrm{R}-$parity of  $\mathrm{I}_1, \ldots, \mathrm{I}_{r-1}$ is also even. Therefore, it results, after applying lemma \ref{lema2}, that the entry order from $H^*_{\sigma^{-1}(1,h)}$ to  $H^*_{\sigma^{-1}(r,h)}$ is conserved. As  $x_1=c^*_{(1,1)}\in H^*_{\sigma^{-1}(1,h)}=H^*_1$ and the entry order is conserved then it is concluded that $x_r=c^*_{(1,\sigma^{-1}(r,h))}$.

The inverse implication can be  obtained by reversing  the process.
  \item [(ii)] The proof is exactly the same as before.
\end{enumerate}

\end{proof}

Notice we have never said $x'_r$ is a boundary point and, in fact, it is not. The idea of the proof is that every extremum of $f^h$ (the boxes) behaves as $f$. The composition ($s$ times) of this approximation  generates something similar to $f^s$ (the points inside the boxes) in the neighborhood  of the extrema (see figure  \ref{fighs}). Then we link an extremum of $f^h$ (at  $x'_r$) with a boundary point of that ``thing similar to $f^s$'' (around $x_r$). Obviously $x'_r$ and $x_r$ belong to the same box generated by the extremum of $f^h$ at $x'_r$; in fact, $\sigma^{-1}_{s}$ moves $x'_r$  to $x_r$ (maximum at $x'_r$) and $\overline{\sigma^{-1}_{s}}$ moves $x'_r$ to $x_r$ (minimum at $x'_r$) because $x'_r$ behaves as the critical point $\mathrm{C}$. Therefore instead of getting information from $x'_r$ (behaving as $\mathrm{C}$, an extremum), we get it from displacement of $x'_r$ (a boundary point) given by $\sigma^{-1}_{s}$.

\begin{ex} Let us pay attention to figures \ref{f3} and  \ref{figure4}. $H_1^*$ is a convex box because the first time this box is   visited the point $\mathrm{C}^*_{(1,12)}=c^*_{(1,1)}$ is reached.  Meanwhile $H_2^*$  and $H_3^*$ are concave  boxes because the first time these boxes are visited  the points $\mathrm{C}^*_{(12,12)}=c^*_{(4,3)}$  and $\mathrm{C}^*_{(8,12)}=c^*_{(4,2)}$ are respectively reached. \end{ex}
As we have said, the boxes, associated to extrema of $f^h$, can be either convex (maximum of $f^h$) or concave  (minimum of $f^h$). The permutation {\large {${\sigma}$}}$^{-1}_{s}$ determines the visiting order in convex boxes. To determine the visiting order in a concave box the permutation  {\large {${\sigma}$}}$^{-1}_{s}$ must be transformed by the same process that transforms a maximum into a minimum: a homotecy, as it can be seen in Appendix. According to this idea we have the conjugated permutation  as shown below.

\begin{defi} Let {\large {${\sigma}$}}$_{p}$ be the cardinal ordering permutation of the $p-$periodic supercycle  of $f$. We define the cardinal ordering conjugated permutation, denoted by  $\overline{{\text{\large ${\sigma}$}}_{p}}$, as \[
\text{\large {$\overline{{\sigma}_{p}}$}}=\begin{pmatrix}1 & \cdots & i & \cdots & p\\
\sigma(p,p)& \cdots & \sigma(p-i+1,p) & \cdots& \sigma(1,p)\end{pmatrix}\ \]
\end{defi}

\begin{defi} Let {\large {${\sigma}$}}$^{-1}_{p}$ be the visiting sequence of the $p$-periodic supercycle  $O_{p}$. We define the conjugated visiting sequence, denoted by  {$\overline{\text{ {\large{$\sigma$}}}^{-1}}_{p}$}, as \[
\overline{\text{\large {${\sigma}$}}^{-1}_{p}}=\begin{pmatrix}1 & \cdots & i & \cdots & p\\
p+1-\sigma^{-1}(1,p)& \cdots & p+1-\sigma^{-1}(i,p) & \cdots& p+1-\sigma^{-1}(p,p)\end{pmatrix}\ \]

\end{defi}

\begin{thm}\label{th1} Let $O_{h}$ and $O_{s}$ be  supercycles  with cardinal points $\{ \mathrm{C}_{(1,h)}^{*},\ldots,\mathrm{C}_{(h, h)}^{*}\}$ and $\{ \mathrm{C}_{(1,s)}^{*},\ldots,\mathrm{C}_{(s, s)}^{*}\}$ respectively. Let  {\large$\sigma$}$^{-1}_h$ and {\large$\sigma$}$^{-1}_s$ be  their visiting sequences and let  $m\in \N$  where \,$1\leq m \leq hs$ and  $m=nh+r, \ 1 \leq r \leq h.$ Then the visiting sequence of $O_{hs}$ is,

 \begin{enumerate}
  \item [i)] if the $\mathrm{R}-$parity of the subsequence preceding $\mathrm{I}_r$ of the $h-$symbolic sequence  is even
  \[ \sigma^{-1}(m,hs)=   (\sigma^{-1}(r,h)-1)s+ \sigma^{-1}(n+1,s)
 \]
  \item  [ii)] if the $\mathrm{R}-$parity of subsequence preceding $\mathrm{I}_r$ of the $h-$symbolic sequence  is odd
  \[ \sigma^{-1}(m,hs)= (\sigma^{-1}(r,h)-1)s+ s+1- \sigma^{-1}(n+1,s)
\]
\end{enumerate}
\end{thm}

\begin{proof}  To know the cardinal point that we are located after $m$ iterations from $\mathrm{C}$, we can write
$f^m(\mathrm{C})=f^{nh}(f^r(\mathrm{C}))$.\\
 As $f^r(\mathrm{C})=x_r\in H^*_{\sigma^{-1}(r,h)},$ using lemma \ref{lema1}, it yields $x_r$ is a boundary point, that is,
$x_r=c^*_{(1,\sigma^{-1}(r,h))}$ or $x_r=c^*_{(s,\sigma^{-1}(r,h))}$. We divide the proof in two steps
\begin{enumerate}
  \item [i)]   The  $\mathrm{R}-$parity of the subsequence preceding $\mathrm{I}_r$ of the $h-$symbolic sequence is assumed to be  even. Therefore, the $\mathrm{R}-$parity of the subsequence preceding to $\mathrm{I}_r$ of the $hs-$symbolic sequence  is also even. Then it followsm by using   lemma \ref{lema2},   that the entry  order from $H^*_{\sigma^{-1}(1,h)}$ to  $H^*_{\sigma^{-1}(r,h)}$ is conserved, as   $x_1=c^*_{(1,1)}$ it results  $x_r=c^*_{(1,\sigma^{-1}(r,h))}$, and applying  proposition \ref{p1}, it yields that $H^*_{\sigma^{-1}(r,h)}$ is a convex box. Given that $f^m(\mathrm{C})=f^{nh}(f^r(\mathrm{C}))=f^{nh}(x_r)$ it results that we have moved $n$ times inside the convex box $H^*_{\sigma^{-1}(r,h)}$ from $x_r=c^*_{(1,\sigma^{-1}(r,h))}$.
      After $n$ movements from $c^*_{(1,\sigma^{-1}(r,h))}$ in a convex box we will be located  at $c^*_{(\sigma^{-1}(n+1,s),\sigma^{-1}(r,h))}$ and, using definition \ref{def4}, it yields  $c^*_{(\sigma^{-1}(n+1,s),\sigma^{-1}(r,h))}= \mathrm{C}^*_{((\sigma^{-1}(r,h)-1)s+ \sigma^{-1}(n+1,s),hs)}$, so  we finally obtain
      \[\sigma^{-1}(m,hs)=(\sigma^{-1}(r,h)-1)s+ \sigma^{-1}(n+1,s)\]

\item  [ii)] The $\mathrm{R}-$parity of the  subsequence preceding $\mathrm{I}_r$ of the $h-$symbolic sequence is assumed to be odd. The proof runs as shown above, the reader should notice that now  $x_r=c^*_{(s,\sigma^{-1}(r,h))}$ because the $\mathrm{R}-$parity of the subsequence preceding $\mathrm{I}_r$ of the $hs-$symbolic sequence  is odd. The entry order is not conserved and $H^*_{\sigma^{-1}(r,h)}$ is a concave box (by
     proposition \ref{p1}).         The  movements inside $H^*_{\sigma^{-1}(r,h)}$  are ruled  by $\overline{\sigma^{-1}_s}$, so after $n$ movements from $c^*_{(s,\sigma^{-1}(r,h))}$ we will be located  at $c^*_{(\overline{\sigma^{-1}}(n+1,s),\sigma^{-1}(r,h))}$ and we obtain, by using definition \ref{def4},
      \[\sigma^{-1}(m,hs)=(\sigma^{-1}(r,h)-1)s+\overline{\sigma^{-1}}(n+1,s)\] Finally taking into account that $\overline{\sigma^{-1}}(n+1,s)=s+1-\sigma^{-1}(n+1,s)$ (see appendix) statement of the theorem is obtained.
\end{enumerate}

\end{proof}

\begin{ex}\label{ejem} We are going to calculate {\large{$\sigma$}}$_3\circ${\large{$\sigma$}}$_4$ $(h=3, s=4)$ to obtain {\large{$\sigma$}}$_{3\cdot 4}= ${\large{$\sigma$}}$_{12}$. According to figures \ref{f1} and \ref{fig2}
\[
\text{\large{$\sigma_{3}^{-1}$}}=\left(\begin{array}{ccc}
1 & 2 & 3\\
1 & 3 &2\end{array}\right)  \qquad \text{\large{$\sigma_{4}^{-1}$}}=\left(\begin{array}{cccc}
1 & 2 & 3&4\\
1 & 4& 3 &2\end{array}\right)  \] that give the symbolic sequences $\mathrm{RLC}$  and $\mathrm{RLLC}$ respectively (see \cite{Mar} for more details).
\begin{itemize}
  \item [a)] If $m=4$ then $m=1\cdot 3+1\equiv n\cdot h+r$. The subsequence preceding $\mathrm{I_r}=\mathrm{I_1}$ is $\mathrm{C}$, therefore the $\mathrm{R}-$parity is even, hence \[\sigma^{-1}(4,12)=(\sigma^{-1}(1,3)-1)4+ \sigma^{-1}(2,4)=4\]that is, we are located at $\mathrm{C}^*_{(4,12)}$ (see figure \ref{f3} and  \ref{figure4}.)

  \item [b)] If $m=5$ then $m=1\cdot 3+2\equiv n\cdot h+r$. The subsequence preceding $\mathrm{I_r}=\mathrm{I_2}$ is $\mathrm{R}$, therefore the $\mathrm{R}-$parity is odd, hence  \[\sigma^{-1}(5,12)=(\sigma^{-1}(2,3)-1)4+ 4+1-\sigma^{-1}(2,4)=9\]that is, we are located at $\mathrm{C}^*_{(9,12)}$ (see figure \ref{f3} and  \ref{figure4}), in fact it is just one movement apart from  $\mathrm{C}^*_{(4,12)},$  obtained when $m=4$.
\end{itemize}

  \end{ex}
\begin{thm}\label{th2} Let  the supercycles $O_{h}$ and $O_{s}$ have cardinal points $\{ \mathrm{C}_{(1,h)}^{*},\ldots,\mathrm{C}_{(h, h)}^{*}\}$ and $\{ \mathrm{C}_{(1,s)}^{*},\ldots,\mathrm{C}_{(s, s)}^{*}\}$ respectively. Let   {\large$\sigma_h$} and {\large$\sigma_s$} be their cardinal ordering permutations.  Then the number of iterates to reach  $\mathrm{C}^*_{(i,hs)}$, with $i=ns+r$, $1\leq i\leq hs$, and $1\leq r\leq s$ is given  by,
 \begin{enumerate}
  \item [i)] if the $\mathrm{R}-$parity of  the subsequence preceding  $\mathrm{I}_{n+1}$ of the $h$-symbolic sequence is even,
  \[\sigma(i,hs)=(\sigma(r,s)-1)h+ \sigma(n+1,h)\]
  \item  [ii)] if the $\mathrm{R}-$parity of  the subsequence preceding $\mathrm{I}_{n+1}$ of the $h$-symbolic sequence is odd,
  \[\sigma(i,hs)=(\overline{\sigma(r,s)}-1)h+ \sigma(n+1,h)\]
\end{enumerate}
\end{thm}
\begin{proof}
Let us suppose that the $\mathrm{R}-$parity of the $h-$sequence preceding $\mathrm{I}_{n+1}$ is even. From definition \ref{def4}
\[\mathrm{C}_{(i,hs)}^{*}=\mathrm{C}_{(ns+r,hs)}^{*}=c_{(r,n+1)}^{*}\]

Therefore $\mathrm{C}_{(i,h)}^{*}$ is in the $r-$th position  of $H^*_{n+1.}$

On one hand,  $\sigma(n+1,h)$ iterates are needed, starting from $\mathrm{C}$, to reach  the box $H^*_{n+1}$, furthermore we will be spotted in $c_{(1,n+1)}^{*}\in H^*_{n+1},$ according lemma \ref{lema2}.

On the other hand,  we need $\sigma(r,s)-1$ iterates inside the box $H^*_{n+1}$ to move from $c_{(1,n+1)}^{*}$ to $c_{(r,n+1)}^{*}$ (by definition of $\sigma(r,s)$ and of  $c_{(-,-)}^{*}$,  see the paragraph following definition \ref{def4}). But to move from $c_{(j,n+1)}^{*}$ to $c_{(j+1,n+1)}^{*}$ the $h$ boxes need to be visited before; therefore, we need $h (\sigma(r,s)-1)$ iterates to move from $c_{(1,n+1)}^{*}$ to $c_{(r,n+1)}^{*}$.

Finally, we obtain that the total number of iterates, to reach $c_{(r,n+1)}^{*}$ from $\mathrm{C}$ is
\[\sigma(i,hs)=(\sigma(r,s)-1)h+ \sigma(n+1,h)\]

Conversely if we suppose that the $\mathrm{R}-$parity of the $h-$sequence preceding $\mathrm{I}_{n+1}$ is odd, then the proof can be handled in much the same way, being the only difference  that  the box   $H^*_{n+1}$ is a concave one. We know that $H^*_{n+1}$ is a concave box by using the $\mathrm{R}-$parity   of the subsequence preceding  $\mathrm{I}_{n+1}$ of the $h-$ symbolic sequence. So we have to use $\overline{\sigma(r,s)}$ instead of $\sigma(r,s)$ to control  movements  inside the box, therefore it results
\[\sigma(i,hs)=(\overline{\sigma(r,s)}-1)h+ \sigma(n+1,h)\]

 The relation between ${\sigma(r,s)}$ and $\overline{\sigma(r,s)}$ can be seen in the appendix.

\end{proof}

\begin{ex}How many iterates are needed to reach $\mathrm{C}_{(9,3\cdot 4)}^{*}$?.  According to example \ref{ejem}, five iterates are needed. Let us check it: if $i=9=2\cdot 4+1\equiv n\cdot s+r$ then the $\mathrm{R}-$parity of the subsequence preceding $\mathrm{I_{n+1}}=\mathrm{I_3}$ of the $3-$symbolic sequence  is odd, hence \[\sigma(9,12)=(\overline{\sigma(1,4)}-1)3+ \sigma(3,3)=(\sigma(4,4)-1)3+ \sigma(3,3)=5\]as it was expected.
 \end{ex}

 Needless to say in the theorems \ref{th1} and \ref{th2} we have only considered  an iteration number less than or equal to $hs$, because full cycles are despised.

\section{Algorithms}
Theorems \ref{th1} and \ref{th2} can be easily written in an algorithmic manner. The goal is direct when we transform theorem \ref{th1} into  an algorithm, because as $m=n\cdot h+r,$ with $1\leq r \leq h$, if $m$ is increased one at a  time then $n$ keeps constant for $h$ iterates and only $r$ varies, therefore $\sigma^{-1}(n+1,s)$ keeps constant for $h$ times and only $\sigma^{-1}(r,h)$ varies. The same works in a similar way for theorem \ref{th2} with $i=n\cdot s+r$ $1\leq r \leq s$, where we work with $\sigma(n+1,h)$ and $\sigma(r,s)$.
\subsection{ Algorithm to obtain  {\large$\sigma^{-1}_{hs}$} from {\large$\sigma^{-1}_h$} and {\large$\sigma^{-1}_s$}}
\begin{enumerate}
  \item Write the $h-$symbolic sequence.
  \item Write $\sigma^{-1}(j,s)$ below every symbol of the $h-$symbolic sequence.
  \item Change $\sigma^{-1}(j,s)$ by $s+1-\sigma^{-1}(j,s)$ if the $R-$parity of the previous sequence of upper symbols is odd.
  \item Add $(\sigma^{-1}(1,h)-1)\cdot s$ to the first term,  $(\sigma^{-1}(2,h)-1)\cdot s$ to the second  term and so successively until you add  $(\sigma^{-1}(h,h)-1)\cdot s$ to the last term.
  \item Repeat consecutively the former sequence ``$ s$'' times, with ``$j$'' changing from $j=1$ to $j=s$.
\end{enumerate}

The sequence obtained in step five is the looked for permutation.

\begin{ex}\label{ejem1} We are going to obtain  {\large$\sigma^{-1}_{3\cdot 4}$} from {\large$\sigma^{-1}_3$} and {\large$\sigma^{-1}_4$}
\begin{enumerate}
  \item \(\mathrm{R} \qquad  \mathrm{L} \qquad \mathrm{C}\)
  \item \[\begin{array}{ccc}
           \mathrm{R} & \mathrm{L} & \mathrm{C} \\
           \sigma^{-1}(j,4) & \quad \sigma^{-1}(j,4) & \quad  \sigma^{-1}(j,4)
         \end{array}
  \]
  \item \[\begin{array}{ccc}

           \sigma^{-1}(j,4) & \quad  s+1-\sigma^{-1}(j,4) & \quad s+1-\sigma^{-1}(j,4)
         \end{array}
  \Longrightarrow \begin{array}{ccc}
                    \sigma^{-1}(j,4) & \ 5-\sigma^{-1}(j,4) & \ 5-\sigma^{-1}(j,4)
         \end{array}\]
  \item
    \begin{multline*}
  \begin{array}{ccc}
                    \sigma^{-1}(j,4) &  5-\sigma^{-1}(j,4) &  5-\sigma^{-1}(j,4)\\
                     +& + & + \\
(\sigma^{-1}(1,3)-1) \cdot 4 & \qquad  (\sigma^{-1}(2,3)-1) \cdot 4 & \qquad (\sigma^{-1}(3,3)-1) \cdot 4
         \end{array}\Longrightarrow \\ \\ \Longrightarrow \begin{array}{ccc}
                  \sigma^{-1}(j,4) &  \quad 13-\sigma^{-1}(j,4) &\quad   9-\sigma^{-1}(j,4)
         \end{array}
 \end{multline*}

  \item

  \begin{multline*}
   {\text{\tiny{$\sigma^{-1}(1,4)$}}}  \  {\text{\tiny{$13-\sigma^{-1}(1,4)$}}}  \   {\text{\tiny{$9-\sigma^{-1}(1,4)$}}} \ | {\text{\tiny{$\sigma^{-1}(2,4)$}}} \  {\text{\tiny{$13-\sigma^{-1}(2,4)$}}} \ {\text{\tiny{$9-\sigma^{-1}(2,4)$}}} |  \\ | \ {\text{\tiny{$\sigma^{-1}(3,4)$}}}  \ {\text{\tiny{$13-\sigma^{-1}(3,4)$}}} \  {\text{\tiny{$9-\sigma^{-1}(3,4)$}}} \ | {\text{\tiny{$\sigma^{-1}(4,4)$}}} \ {\text{\tiny{$13-\sigma^{-1}(4,4)$}}} \ {\text{\tiny{$9-\sigma^{-1}(4,4)$}}} \ \Longrightarrow \\
   \Longrightarrow  1 \quad 12 \quad 8\quad 4 \quad 9 \quad 5 \quad 3 \quad 10 \quad 6 \quad 2 \quad 11 \quad 7
  \end{multline*}


\end{enumerate}  It can be checked in figures \ref{f3} and  \ref{figure4} that we have obtained {\large$\sigma^{-1}$}$_{3\cdot 4}$ (the visiting sequence of $O_{3\cdot4}$)\end{ex}

\subsection{ Algorithm to obtain  {\large$\sigma_{hs}$} from {\large$\sigma_h$} and {\large$\sigma_s$}}
\begin{enumerate}
   \item Write ``$s$''  times $\sigma(j,h)$.
    \item Add $(\sigma(1,s)-1)\cdot h$ to the first term,  $(\sigma(2,s)-1)\cdot h$ to the second  term and so successively until you add  $(\sigma(s,s)-1)\cdot h$ to the last term
       \item Repeat consecutively the former sequence ``$h$'' times, with ``$j$'' changing from $j=1$ to $j=h$.
        \item If the $R-$parity of the subsequence preceding $I_j$ of the $h-$symbolic sequence  is odd then $\sigma(i,s)$ must change by $\overline{\sigma(i,s)}$ for $i=1,\ldots,s$.

\end{enumerate}

The obtained  sequence is the looked for permutation.

\begin{ex}We are going to obtain  {\large$\sigma$}$_{3\cdot 4}$ from {\large$\sigma$}$_3$ and {\large$\sigma$}$_4$
\begin{enumerate}
    \item $\begin{array}{cccc}
             \sigma(j,3)  &   \quad \sigma(j,3) &  \quad \sigma(j,3) &  \quad \sigma(j,3)
           \end{array}   $\\
  \item \[\begin{array}{cccc}
           \sigma(j,3) & \sigma(j,3) & \sigma(j,3)& \sigma(j,3) \\
           +&+&+&+\\(\sigma(1,4)-1)\cdot 3& \qquad(\sigma(2,4)-1)\cdot 3 & \qquad (\sigma(3,4)-1)\cdot 3& \qquad (\sigma(4,4)-1)\cdot 3
                  \end{array} \]

  \item
 \begin{multline*}
    \begin{array}{cccc}
           \sigma(1,3) & \sigma(1,3) & \sigma(1,3)& \sigma(1,3) \\
           +&+&+&+\\(\sigma(1,4)-1)\cdot 3& \qquad(\sigma(2,4)-1)\cdot 3 & \qquad (\sigma(3,4)-1)\cdot 3& \qquad (\sigma(4,4)-1)\cdot 3
                  \end{array} \\ \\ \\ \begin{array}{cccc}
           \sigma(2,3) & \sigma(2,3) & \sigma(2,3)& \sigma(2,3) \\
           +&+&+&+\\(\sigma(1,4)-1)\cdot 3& \qquad(\sigma(2,4)-1)\cdot 3 & \qquad (\sigma(3,4)-1)\cdot 3& \qquad (\sigma(4,4)-1)\cdot 3
                  \end{array} \\ \\ \\
                  \begin{array}{cccc}
           \sigma(3,3) & \sigma(3,3) & \sigma(3,3)& \sigma(3,3) \\
           +&+&+&+\\(\sigma(1,4)-1)\cdot 3& \qquad(\sigma(2,4)-1)\cdot 3 & \qquad (\sigma(3,4)-1)\cdot 3& \qquad (\sigma(4,4)-1)\cdot 3
                  \end{array}
 \end{multline*}

 \item
 \begin{multline*}
    \begin{array}{cccc}
           \ \sigma(1,3) & \sigma(1,3) & \sigma(1,3)& \sigma(1,3) \\
           +&+&+&+\\(\sigma(1,4)-1)\cdot 3& \qquad(\sigma(2,4)-1)\cdot 3 & \qquad (\sigma(3,4)-1)\cdot 3& \qquad (\sigma(4,4)-1)\cdot 3
                  \end{array} \\ \\ \\ \begin{array}{cccc}
           \sigma(2,3) & \sigma(2,3) & \sigma(2,3)& \sigma(2,3) \\
           +&+&+&+\\(\overline{\sigma(1,4)}-1)\cdot 3& \qquad(\overline{\sigma(2,4)}-1)\cdot 3 & \qquad (\overline{\sigma(3,4)}-1)\cdot 3& \qquad (\overline{\sigma(4,4)}-1)\cdot 3
                  \end{array} \\ \\ \\
                  \begin{array}{cccc}
           \sigma(3,3) & \sigma(3,3) & \sigma(3,3)& \sigma(3,3) \\
           +&+&+&+\\(\overline{\sigma(1,4)}-1)\cdot 3& \qquad(\overline{\sigma(2,4)}-1)\cdot 3 & \qquad (\overline{\sigma(3,4)}-1)\cdot 3& \qquad (\overline{\sigma(4,4)}-1)\cdot 3 \end{array}\\
                      \Longrightarrow   \begin{array}{cccc}
           1 & \quad 10 & \quad 7& \quad 4 \\
                             \end{array} \  \qquad  \ \begin{array}{cccc}
           6& \quad  9 & \quad 12& \quad 3                           \end{array} \ \qquad  \
                \begin{array}{cccc}
           5 &\quad  8 & \quad 11& \quad 12
           \end{array}\\
                       \end{multline*}

\end{enumerate}  It can be checked  that  the  obtained permutation {\large$\sigma$}$_{3\cdot 4}$  is the inverse permutation of {\large$\sigma^{-1}$}$_{3\cdot 4}$ obtained in example \ref{ejem1}.  Observe that the added sequence $\sigma(i,s)$ is always the same or a mere reflection of it (see appendix) so the algorithm is   very fast to run.
\end{ex}

\section{Appendix}
\subsection{Meaning of the conjugation for permutations {\large$\sigma_p$} and {\large$\sigma^{-1}_p$}}

Let us consider two unimodal functions, with a maximum and a minimum respectively, such that they have the same orbits.   It is well known that their symbolic sequences are conjugated by the change $\mathrm{R}\longleftrightarrow \mathrm{L}$. The conjugation can be seen as a homotecy that transforms the maxima into a minima and viceversa. If we pay attention to figure \ref{figapendice} we can see as the sequence $\mathrm{RLC}$ (in the maximum) is transform into $\mathrm{LRC}$ (in the minimum). We can interprete geometrically the homotecy as a reflection about the  dashed line. We want to know  how the reflection affects to  {\large$\sigma_p$} and {\large$\sigma^{-1}_p$}.

\begin{figure}
\includegraphics[width=0.60\textwidth]{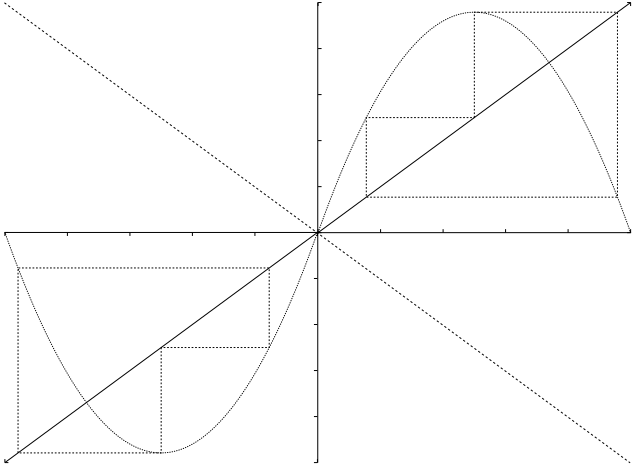}
\caption{\label{figapendice}Orbit, with symbolic sequence $\mathrm{RLC}$, in the first quadrant is transformed into an orbit with symbolic sequence $\mathrm{LRC}$ in the third quadrant by means of a homotecy. The homotecy is performed as a reflection about the dashed line.}
\end{figure}

\begin{enumerate}
  \item[-]   The permutation
  \begin{equation}\label{ec2}
\text{\large {${\sigma}_{p}$}}=\begin{pmatrix}1 & \cdots & i & \cdots & p\\
\sigma(1,p)& \cdots & \sigma(i,p) & \cdots& \sigma(p,p)\end{pmatrix}\
  \end{equation} is defined by $f^{\sigma(i,p)}=C^*_i$ (see \cite{Mar}) therefore, its reflection is given by

  \begin{equation}\label{ec3}
\text{\large {$\overline{{\sigma}_{p}}$}}=\begin{pmatrix}1 & \cdots & i & \cdots & p\\
\sigma(p,p)& \cdots & \sigma(p-i+1,p) & \cdots& \sigma(1,p)\end{pmatrix}\
  \end{equation}
  that is,

  \[\left(
      \begin{array}{cccc}
        0 & \cdots &0& 1 \\
        0 & \ldots &1& 0 \\
        \vdots & \ddots &\vdots& 0 \\
        1&\ldots& \ldots&0
      \end{array}
    \right)\left(
             \begin{array}{c}
               \sigma(1,p) \\
               \sigma(2,p) \\
               \vdots \\
               \sigma(p,p) \\
             \end{array}
           \right)=\left(
             \begin{array}{c}
               \sigma(p,p) \\
               \sigma(p-1,p) \\
               \vdots \\
               \sigma(1,p) \\
             \end{array}
           \right)
  \] where the matrix  $\left(
      \begin{array}{cccc}
        0 & \cdots &0& 1 \\
        0 & \ldots &1& 0 \\
        \vdots & \ddots &\vdots& \vdots \\
        1&\ldots& \ldots&0
      \end{array}
    \right)$ accomplishes the reflection.

  \item[-]  The permutation {\large$\sigma^{-1}_p$} determines how the $p-$periodic orbit is visited (see \cite{Mar}), that is, what cardinal points are visited one after another. The reflection of {\large$\sigma^{-1}_p$} is {\large$\overline{\sigma^{-1}_p}$} that means the orbit is gone through in reverse order. Therefore one of the permutations produces a shift which is cancelled by the other, hence

      \[\sigma^{-1}(i,p)+ \overline{\sigma^{-1}(i,p)}=p+1\]

      {\large$\sigma^{-1}_p$}  and {\large$\overline{\sigma^{-1}_p}$} are deduced from (\ref{ec2}) and (\ref{ec3})
\end{enumerate}


\section{Conclusions}

To understand the language of chaos we need to understand its grammar. A grammar, a set of rules, that determines how to combine some structures with others to generate new ones. In nonlinear dynamics basic structures are orbits, that are characterized by their cardinal ordering permutations (COP) ~\cite{Mar}; whereas, grammar is the COP composition rules.

 These COP composition rules, unknown until  now, have been stated  in this paper  in two theorems. Such composition rules have two importants and immediates implications:
 \begin{itemize}
   \item [(i)]  Given the COP of a $h-$periodic and $s-$periodic orbit then the COP of $hs-$periodic orbit is deduced, that is, the relative position of all points belonging to the $hs-$periodic orbit and its visiting order. Furthermore using the composition law in a recursive way it is possible to obtain COP for $h^{a}s^b-$periodic orbits $a,b\in \mathbb{Z^+}$, that is, with a couple of COP we can determine infinitely many COP and from them to fully understand  the associated orbits.
   \item [(ii)] As COP for periodic doubling cascade is known \cite{Mar} we would immediately obtain COP for both periodic doubling cascade and saddle-node bifurcation cascade \cite{Mar1, Mar2} with any arbitrary  basic period $(p\cdot 2^k)$.
 \end{itemize}

 Known the points $(i)$ and $(ii)$ all underlying periodic structures in an unimodal function $f$ are determined, hence the importance of the theorems \ref{th1} and \ref{th2}.

We want to point out that the COP for the period doubling cascade can be obtained by using the theorems proven in this paper. It is enough considering $\sigma_h=\sigma_s \ (h=s=2)$ to obtain $\sigma_{h.s}=\sigma_{2\cdot 2}$
By using recursively this process  we get the COP for the $2^n-$periodic orbit of the cascade. No matter whether we take $h=2^{(n-1)}$ and $s=2$ or $h=2$ and $s=2^{(n-1)};$ in this case the process is absolutely symmetric.
But if we use the second one we get a new proof of theorem 1 in \cite{Mar}.

 To prove the theorems we  needed to know if some points of $h-$periodic orbit were associated with maxima or minima of $f^h$. Since it would be wholly inoperative to get that information by deriving $f^h$ when $h$ is big, two steps were carried on: a) to provide this information implicitly in the theorems, b) to obtain this information in an simple way, even if $h$ is very big (proposition \ref{p1}). The solutions to these problems play a pivotal role in proving those  theorems. We must point out that analytical information has been obtained from topological one (what authors named ``subsequence preceding''), establishing a link between two apparently unconnected fields, that must be studied more carefully. It is well known that, when two apparently unconnected mathematical fields turn out to describe the same thing, but from two different points of view, the problems they address are very largely simplified.

--------------------------------------------


\begin{thebibliography}{10}
\bibitem{Hao} \textsc{ Hao Bai-Lin }: \textit{Elementary Symbolic dynamicss}. World Scientific 1989.
\bibitem{metro} \textsc{ N. Metroplois, M. L. Stein, and P. R. Stein}: \textit{On Finite Limit Sets for Transformations on the Unit Interval}, J. Combinatorial Theory \textbf{15} (1973) 25-44
 \bibitem{De} \textsc{B. Derrida, A.  Gervois, Y. Pomeau }: \textit{Iteration of endomorphisms on the rael axes and representation of numbers}, L'I.H.P. section A   \textbf{29} 3, (1978) 305-356
\bibitem{Mar} \textsc{ J. San Martín, Mª José Moscoso, A. González Gómez}: \textit{The universal cardinal ordering of fixed points}, Chaos, Solitons \& Fractals, \textbf{42} (2009) 1996-2007
  \bibitem{Fei1} \textsc{M. J. Feigenbaum}: \textit{Quantitative Universality for a Class of Nonlinear
Transformations}, J. Stat. Phys. \textbf{19} (1978) 25-52

\bibitem{Fei2} \textsc{M. J. Feigenbaum}: \textit{The Universality Metric Properties for  Nonlinear
Transformations}, J. Stat. Phys. \textbf{21} (1979) 669-706
 \bibitem{Mil} \textsc{J. Milnor, W. Thurston}: \textit{On iterated maps of the interval}, Dinamical systems. Lecture notes in math, vol. 1342. Berlin: Springer;  (1988) 465-563.

 \bibitem{Gil} \textsc{R. Gilmor, M. Lefranc}: \textit{The topology of chaos},  New York: Willye;  (2002).

 \bibitem{Mar1} \textsc{ J. San Martín, D. Rodríguez-Pérez}: \textit{Intermittency cascade}, Chaos, Solitons \& Fractals, \textbf{32}(2) (2007) 816-831
      \bibitem{Mar2} \textsc{ J. San Martín, D. Rodríguez-Pérez}: \textit{Conjugation of cascades}, Chaos, Solitons \& Fractals, \textbf{39}(2) (2009) 666-681
 \bibitem{Glen} \textsc{Glen R. Hall}: \textit{Some examples of permutations modellings area preserving monotone twist map}, Physica D: Nonlinear Phenomena, 28  (1978) 393-400


\bibitem{A} \textsc{J. M. Amigó, M. B. Kennel, L.  Kocarev }: \textit{The permutations entropy rate equials the metric entropy rate for ergodic information sources and ergodic dynamical systems}, Physica D 210(2005) 77-95.
    \bibitem{Am} \textsc{J. M. Amigó, L.  Kocarev, I. Tomovski }: \textit{Discrete entropy }, 228(2007) 77-85
    \bibitem{Ami} \textsc{J. M. Amigó, M. B. Kennel}: \textit{Topological permutations entropy }, Physica D 231(2007) 137-142.
  \bibitem{Ban} \textsc{C. Bandt, G. Keller, B.  Pompe }: \textit{Entropy of interval maps via permutations}, Nonlinearity  15(2002) 1595-1602


     \bibitem{Amig} \textsc{J. M. Amigó, M. B. Kennel}: \textit{Forbbiden ordinal patterns in higher dimensional dynamics}, Physica D 237(2008) 2893-2899.
           \bibitem{Amigo} \textsc{J. M. Amigó, L.  Kocarev, J. Szczepanski }: \textit{Order pattern and chaos}, Physics Letters A 355(2006) 27-31.






 \bibitem{Z} \textsc{L. Zunino, M. Zanin, B. M. Tabak, D.G. Pérez}: \textit{Forbbiden  patterns, permutation entropy and stock market inefficiency}, Physica A 388(2009) 2854-2864.


\end{thebibliography}
\end{document}